\pdfoutput=1
\newif\ifappendix\appendixfalse
\newif\ifbw\bwtrue
\RequirePackage[usenames,dvipsnames]{xcolor}
\PassOptionsToPackage{ocgcolorlinks}{hyperref}
\documentclass{easychair}
\usepackage[numbers,sort&compress]{natbib}
\usepackage[english]{babel}
\usepackage[T1]{fontenc}
\usepackage[utf8]{inputenc} 
\usepackage{microtype}
\usepackage{wrapfig}
\usepackage{upgreek}
\usepackage{paralist}
\usepackage{enumitem}
\usepackage{stmaryrd}
\usepackage{graphicx}
\usepackage{xspace,xcolor}
\usepackage{tabularx}

\usepackage{xspace}
\usepackage{nicefrac}
\usepackage[labelformat=empty,margin=0pt]{subfig}
\expandafter\def\expandafter\UrlBreaks\expandafter{\UrlBreaks
  \do\a\do\b\do\c\do\d\do\e\do\f\do\g\do\h\do\i\do\j%
  \do\k\do\l\do\m\do\n\do\o\do\p\do\q\do\r\do\s\do\t%
  \do\u\do\v\do\w\do\x\do\y\do\z\do\A\do\B\do\C\do\D%
  \do\E\do\F\do\G\do\H\do\I\do\J\do\K\do\L\do\M\do\N%
  \do\O\do\P\do\Q\do\R\do\S\do\T\do\U\do\V\do\W\do\X%
  \do\Y\do\Z}

\usepackage{tikz}
\usetikzlibrary{calc}

\usepackage{amsmath,amsfonts,amsthm,amssymb,mathtools}
\usepackage[ruled,vlined]{algorithm2e}
\usepackage[bb=boondox]{mathalfa}
\usepackage{nicefrac}
\usepackage{bbm}

\newcommand{\EasyCrypt}{\textsf{EasyCrypt}\xspace}

\newcommand{\Sprhl}{\textsf{pRHL}\xspace}

\usepackage{mathpartir}

\usepackage{todonotes}

\newcommand{\rname}[1]{[{\sc #1}]}
\usepackage{xargs}
\newcommandx{\note}[2][1=]{\todo[inline,linecolor=red,backgroundcolor=red!25,bordercolor=red,#1]{#2}}

\newcommand{\ExpD}{\mathbb{E}}

\newcommand{\dnull}[1][]{{{\mathbb{0}}^{#1}}}
\newcommand{\dunit}[2][]{{{\delta}^{#1}_{#2}}}
\newcommand{\dlet}[3]{\ExpD_{{#1} \sim {#2}} [{#3}]}

\newcommand{\Dist}{\ensuremath{\mathbb{D}}}

\newcommand{\pr}[2]{\mathsf{Pr}_{#1} \left[#2 \right]}
\newcommand{\wt}[1]{|#1|}
\newcommand{\supp}{\mathsf{supp}}

\newcommand{\True}{\ensuremath{\mathop{\top}}}

\newcommand{\ptrue}{1}
\newcommand{\pfalse}{0}

\newcommand{\subst}[2]{[#1 := #2]}
\newcommand{\bernD}{\mathbf{Bern}}

\newcommand{\ZZ}{\mathbb{Z}}
\newcommand{\RR}{\mathbb{R}}

\newcommand{\Zn}[1]{\ZZ/\raisebox{-0.5ex}{${#1}\ZZ$}}
\newcommand{\T}{\mathbf{T}}
\newcommand{\Var}{\mathcal{X}}

\newcommand{\Mem}{\mathbf{State}}
\newcommand{\PMem}{\Dist(\Mem)}

\newcommand{\denot}[1]{\llbracket #1 \rrbracket}
\newcommand{\dsem}[2]{\denot{#2}_{#1}}

\newcommand{\selfcomp}[2]{{#1^{\langle #2 \rangle}}}

\newcommand{\proj}{\pi}

\newcommand{\kw}[1]{\ensuremath{\mathbf{#1}}}

\newcommand{\kif}{\kw{if}}
\newcommand{\kthen}{\kw{then}}
\newcommand{\kelse}{\kw{else}}
\newcommand{\kwhile}{\kw{while}}

\newcommand{\kfor}{\kw{for}}
\newcommand{\kdo}{\kw{do}}
\newcommand{\kto}{\kw{to}}

\newcommand\hoare[3]{\{#1\}\;#2\;\{#3\}}

\newcommand{\rnd}{%
  \stackrel{\raisebox{-.15ex}[.25ex]{\tiny %
  $\mathdollar$}}{\raisebox{-.2ex}[.2ex]{$\leftarrow$}}}
\newcommand{\asn}{%
  \stackrel{}{\raisebox{-.1ex}[.2ex]{$\leftarrow$}}}

\newcommand{\iass}[2]{#1 \leftarrow #2}
\newcommand{\irnd}[2]{#1 \rnd #2}
\newcommand{\icondT}[2]{\kif\ {#1}\ \kthen\ {#2}}
\newcommand{\icond}[3]{\kif\: #1\: \kthen\: #2\: \kelse\: #3}
\newcommand{\iwhile}[2]{\kwhile\: #1 \: \kdo \: #2}
\newcommand{\ifor}[4]{\kfor\: #1 \: = #2\: \kto \: #3\: \kdo \: #4}
\newcommand{\iforin}[3]{\kfor\: {#1} \in {#2}\: \kdo \: {#3}}

\newcommand{\iskip}{\kw{skip}}
\newcommand{\iabort}{\kw{abort}}

\newcommand{\itfor}[2]{{#1}^{[#2]}}

\newcommand{\expr}{\mathcal{E}}
\newcommand{\dexpr}{\mathcal{D}}

\newcommand{\eqdef}{\mathrel{\stackrel{\scriptscriptstyle \triangle}{=}}}

\newcommand\q{[\![}
\newcommand\p{]\!]}

\newcommand{\sidel}{\langle 1\rangle}
\newcommand{\sider}{\langle 2\rangle}

\newcommand{\Equiv}[4]{%
  \vDash {#1} \sim {#2}~:~ {#3} \Longrightarrow ~{#4} }

\newcommand{\pre}{\Phi}
\newcommand{\post}{\Psi}

\newcommand{\EqMem}{\mathsf{EqMem}}
\newcommand{\Eqmem}[2]{\EqMem^{\langle{#1}\rangle, \langle{#2}\rangle}}

\usepackage{listings}

\usepackage[scaled]{beramono}
\newcommand{\Small}{\fontsize{8.2pt}{8.4pt}\selectfont}
\newcommand*{\LSTfont}{\Small\ttfamily\SetTracking{encoding=*}{-60}\lsstyle}

\newcommand{\lstrnd}{\stackrel{\raisebox{-.15ex}{\ensuremath{\scriptscriptstyle\$}}}{\raisebox{-.2ex}{\ensuremath{\leftarrow}}}}

\newcommand{\coupledsupp}[4]{%
  {#1} \blacktriangleleft_{#4} \langle {#2} \mathrel{\&} {#3} \rangle}

\newcommand{\eqsem}[3]{#1 \vdash #2 \equiv #3}

\let\xinferrule\inferrule
\renewcommand{\inferrule}[3][]{%
  \ifx&#1&
    \xinferrule*{#2}{#3}
  \else
    \xinferrule*[right=\rname{#1}]{#2}{#3}
  \fi}

\usepackage{hyperref}
\usepackage[capitalize]{cleveref}

\theoremstyle{plain}
\newtheorem{theorem}{Theorem}
\newtheorem{prop}[theorem]{Proposition}
\newtheorem{lemma}[theorem]{Lemma}

\newtheorem{fact}[theorem]{Fact}
\newtheorem{definition}[theorem]{Definition}
\newtheorem{example}[theorem]{Example}

\begin{document}

\title{Proving uniformity and independence \\ by self-composition and coupling}

\titlerunning{Uniformity and independence by couplings}
\author{Gilles Barthe\inst{1} \and Thomas Espitau\inst{2} \and 
  Benjamin Grégoire\inst{3} \\ Justin Hsu\inst{4} \and
  Pierre-Yves Strub\inst{5}}
\authorrunning{G. Barthe, T. Espitau, B. Grégoire, J. Hsu, P.-Y. Strub}
\institute{ $\mbox{}^1$ IMDEA Software Institute \qquad
  $\mbox{}^2$ Sorbonne Universit\'es, UPMC Paris 6 \qquad\\
  $\mbox{}^3$ Inria \qquad
  $\mbox{}^4$ University of Pennsylvania \qquad
  $\mbox{}^5$ École Polytechnique}
\maketitle
\begin{abstract}
\emph{Proof by coupling} is a classical proof technique for establishing
probabilistic properties of two probabilistic processes, like stochastic dominance and rapid mixing of
Markov chains. More recently, couplings have been investigated as a useful abstraction
for formal reasoning about relational properties of probabilistic
programs, in particular for modeling reduction-based cryptographic
proofs and for verifying differential privacy. In this paper, we
demonstrate that probabilistic couplings can be used for verifying
\emph{non-relational} probabilistic properties.  Specifically, we show that the
program logic \Sprhl---whose proofs are formal versions of proofs by
coupling---can be used for formalizing uniformity and probabilistic
independence. We formally verify our main examples using the
\textsf{EasyCrypt} proof assistant.
\end{abstract}
\section{Introduction}

Uniformity and probabilistic independence are two of the most useful and
commonly encountered properties when analyzing randomized computations. Uniform
distributions are a central building block of randomized algorithms. Arguably
the simplest non-trivial distribution---the coin flip---is a uniform distribution
over two values. Given access to uniform samples, there are known
transformations for converting the samples to simulate more complex
distributions, like Gaussian or Laplacian distributions. Conversely, turning 
samples from various non-uniform distributions into uniform
samples is an active area of research.

Probabilistic independence is no less useful. The probability of a conjunction
of independent events can be decomposed as a product of probabilities of
individual events, each which can then be analyzed in isolation.  Independent
random variables are also needed to apply more sophisticated mathematical tools,
like concentration inequalities. 

Given these and other applications, it is not surprising that
researchers have investigated different methods of reasoning about uniformity
and independence. For instance, \citet{PearlP86} develop
an axiomatic theory based on \emph{graphoids} for modeling conditional
independence in probability theory. However, proving uniformity and
independence by program verification remains a challenging task. Most
verification techniques for probabilistic programs do not treat these
properties as first-class assertions, and rely on reasoning principles
that are cumbersome to use. Often, the only way to prove uniformity or
independence is to prove exact values for the probability of specific events.

For example, consider a formal system for proving properties of the form
$\pr{\dsem{m}{s}}{E}=p$, which capture the fact that the event $E$ has
probability $p$ in the distribution obtained by executing the randomized program
$s$ on some initial memory $m$ (many existing systems use this idea,
e.g.~\citep{Kozen85,Morgan96,Ramshaw79,Hartog:thesis,RandZ15,ChadhaCMS07}).
Suppose that we want to prove that a program variable $x$ of some finite type
$A$ is uniformly distributed in the output distribution $\dsem{m}{s}$. The only
way to show this property is to analyze the probability of each output: for
every $a\in A$, prove that $\pr{\dsem{m}{s}}{x=a}=\frac{1}{|A|}$.

For independence, the situation is similar. Assume
that we want to prove that the two program variables $x$ and $y$ of
respective types $A$ and $B$ are (probabilistically) independent
in the output distribution $\dsem{m}{s}$. This can be done by
exhibiting functions $f,g,h$ such that for every $a\in A$ and $b\in
B$, we have: $\pr{\dsem{m}{s}}{x=a}=f(a)$, $\pr{\dsem{m}{s}}{y=b} =
g(b)$, $\pr{\dsem{m}{s}}{x=a\land y=b}=h(a,b)$. Then, independence
between $x$ and $y$ holds by proving that $h(a,b)=f(a)\cdot g(b)$ for every $a\in
A$ and $b\in B$.

While these approaches work in theory, they can be laborious in practice. It
may be awkward to express the probability of $x = a$, and the
functions $f$, $g$ and $h$ may be difficult to produce. The main
contribution of this paper is an alternative method based on
probabilistic couplings for proving uniformity and independence.
Probabilistic couplings are a classical method for proving
sophisticated probabilistic properties (e.g., stochastic dominance,
rapid mixing of Markov chains, and
more~\citep{Thorisson00,Lindvall02,Villani08}). More recently,
couplings have been used to reason about relational properties of
probabilistic programs, notably differential
privacy~\citep{BEGHSS15,BGGHS16}. Here we show that uniformity and
independence properties can also be verified using coupling, despite
being \emph{non-relational} properties. As a consequence, our verification
method inherits the many advantages of reasoning by couplings:
compositional reasoning, and no need to reason directly about
probabilistic events. Concretely, we show how uniformity and independence can be
captured in the relational program logic \Sprhl~\citep{BartheGZ09}.

In summary, our main contributions are novel methods to prove
uniformity and independence properties of probabilistic programs. We
prove the soundness of the methods and demonstrate their usefulness on
a class of case studies.

\subsection*{Detailed Contributions}

\paragraph*{Uniformity.}
Suppose we have a program $s$ with a program variable $x$ ranging over
a finite set $A$, and we want to show that $x$ is distributed
uniformly over $A$ after executing $s$. Rather than computing the
probability of $\pr{\dsem{m}{s}}{x = a}$ for each $a \in A$, it
suffices to show that the probabilities of any two outputs are equal:
\[
  \forall a_1, a_2 \in A.\ \pr{\dsem{m}{s}}{x = a_1} =
  \pr{\dsem{m}{s}}{x = a_2}.
\]
Now, we can view uniformity as a relational property: if we consider
two runs of $s$, then the probability of $x$ being $a_1$ in the first
run should be equal to the probability of $x$ being $a_2$ in the
second run. In \Sprhl, this property is described by the following
judgment:
\[
  \forall a_1, a_2 \in A.\ \Equiv{s}{s}{\phi}{x\sidel = a_1 \iff x\sider=a_2}
\]
where the assertion $\phi$ asserts that the initial states are equal.

\paragraph*{Independence.}
Proving probabilistic independence is more involved. We show how to prove
independence in two different ways. Assume that we want to prove that
the program variables $x$ and $y$ of respective finite types $A$ and
$B$ are independent. First, if the distribution of $\langle
x,y\rangle$ is uniformly distributed over $A\times B$, then $x$ and
$y$ are independent (and are themselves uniformly
distributed). Indeed, assume that for all $a\in A$ and $b\in B$ we
have $\pr{\dsem{m}{s}}{x = a \land y = b} = \frac{1}{|A|\cdot
  |B|}$. Then we have $\pr{\dsem{m}{s}}{x = a} = \sum_{b\in B}
\pr{\dsem{m}{s}}{x = a \land y = b} = \frac{1}{|A|}$. A similar
argument applies to the probability that $y=b$, from which
independence follows. Thus, our first method of proving independence
is by reduction to proving uniformity.

This approach is simple to use, but it only applies to proving independence of
uniform random variables.
A more expressive, but also slightly more complicated approach is to
express probabilistic independence as a property of a
modified version of the program, without any requirement on
uniformity. More specifically, independence of $x$ and $y$ can be
derived from the equality between the probabilities of $x=a\land y=b$
and $x_1=a \land y_2=b$, where in the first case the probability is
taken over the output of the original program $s$, and in the second case the
probability is taken over the output of the program $s_1;s_2$, where
$s_1$ and $s_2$ are renamings of $s$ (we call $s_1;s_2$ a
\emph{self-composition} of $s$~\citep{BartheDR04,DarvasHS05}). The
reason is not hard to see. Since the composed programs operate on
disjoint memory, the final combined output distribution models two
independent runs of the original program $s$. So, the probability
$\pr{\dsem{m_1\uplus m_2}{s_1;s_2}}{x_1=a \land y_2=b}$---where $m_1\uplus m_2$ is the
disjoint union of two copies of $m$---is
equal to the product of $\pr{\dsem{m_1}{s_1}}{x_1=a}$ and
$\pr{\dsem{m_2}{s_2}}{y_2=b}$. Since $s_1$ and $s_2$ are just renamed versions
of the
original program $s$, these probabilities are in turn equal to
$\pr{\dsem{m}{s}}{x=a}$ and $\pr{\dsem{m}{s}}{y=b}$ in the original
program.

Our encoding casts independence as a relational property between a
program $s$ and its self-composition $s_1;s_2$, a property which can
be directly expressed in \Sprhl:
\[
   \forall a \in A, b \in B.\
   \Equiv{s}{s_1;s_2}{\phi}{(x\sidel = a\land y\sidel = b)
     \iff (x_1\sider=a \land y_2\sider=b)}
\]
where the precondition $\phi$ captures the initial conditions. We show
that our approach extends to independence and conditional independence
of sets of program variables.

\paragraph*{Outline}
\cref{sec:background} and \cref{sec:setting} provide the
relevant mathematical background and introduce the setting of our
work. \cref{sec:unif},~\cref{sec:indep} and~\cref{sec:cindep}
respectively address the case of uniformity, independence, and
conditional independence. In each case we demonstrate our method using
classic examples of randomized algorithms. We conclude
the paper with a discussion of alternative techniques for verifying
these properties.

\section{Mathematical Background}\label{sec:background}
For the sake of simplicity, we restrict ourselves to discrete (countable)
sub-distributions.

\begin{definition}
  A \emph{sub-distribution} over a set $A$ is defined by a mass function
  $\mu : A \to \RR^+$, which gives the probability of the unitary events $a \in
  A$. This mass function must be s.t.  $\sum_{a \in A} \mu(a)$ is well-defined
  and its \emph{weight} satisfies
  $\wt{\mu} \eqdef \sum_{a\in A} \mu(a) \leq 1$.
  In particular, the \emph{support} of the sub-distribution
  $\supp(\mu) \eqdef \{ a \in A \mid \mu(a) \neq 0 \}$
  is discrete.
  When $\wt{\mu}$ is equal to $1$, we call $\mu$ a \emph{distribution}.  We let
  $\Dist(A)$ denote the set of sub-distributions over $A$.
  An \emph{event} over $A$ is a predicate over $A$. The probability of an
  event $E$ in a sub-distribution $\mu$, written $\pr{x \sim \mu}{E}$,
  is defined as
  $\sum_{\{x \in A \mid E(x)\}} \mu(x)$.
\end{definition}
When working with sub-distributions over tuples, the probabilistic
versions of the usual projections on tuples are called
\emph{marginals}. For distributions over pairs, we define the
\emph{first} and \emph{second marginals} $\proj_1(\mu)$ and
$\proj_2(\mu)$ of a distribution $\mu$ over $A\times B$ by $\proj_1
(\mu)(a)\eqdef\sum_{b\in B} \mu(a,b)$ and $\proj_2 (\mu)(b)\eqdef\sum_{a\in A}
\mu(a,b)$.  
We are now ready to formally define coupling.
\begin{definition}
  Let $A_1$ and $A_2$ be two sets, and let $\post\subseteq A_1\times
  A_2$. A $\post$-\emph{coupling} for two sub-distributions $\mu_1, \mu_2$
  resp.\ over $A_1$ and $A_2$ is a sub-distribution $\mu \in \Dist(A_1 \times
  A_2)$ such that $\proj_1(\mu)=\mu_1$ and $\proj_2(\mu)=\mu_2$ and
  $\supp (\mu) \subseteq\post$. We write $\coupledsupp {}{\mu_1}
  {\mu_2} \post$ to denote the existence of a $\post$-coupling.
\end{definition}
In addition to the general definition, we shall also consider a
special case of coupling: specifically, we say that $(\mu_1,\mu_2)$
are $f$-\emph{coupled} if $f:A_1\rightarrow A_2$ is a bijection such that
$\mu_1(x)=\mu_2(f(x))$ for every $x\in A_1$. In this case, we write
$\coupledsupp f {\mu_1} {\mu_2} {}$.

Previous works establish a number of basic facts about couplings, see
e.g.~\citet{BartheGZ09,BEGHSS15}, In particular, one useful consequence of
couplings is that they can show that one event has smaller probability than
another.

\begin{lemma}[Fundamental lemma of coupling] \label{l:cplfd}
Let $E_1$ and $E_2$ be predicates over $A_1$ and $A_2$, and let $\post
\eqdef \{ (x_1,x_2) \mid (x_1\in E_1) \Rightarrow (x_2\in E_2) \}$. If
$\coupledsupp {}{\mu_1} {\mu_2} \post$, then $\pr{x_1 \sim \mu_1}{E_1}
\leq \pr{x_2 \sim \mu_2}{E_2}$.
\end{lemma}
One can immediately derive a variant of the lemma where $\iff$ and
$=$ are used in place of $\Rightarrow$ and $\leq$ respectively.  The following
lemma provides a converse to the fundamental lemma of coupling in the special
case where we are interested in proving the equality of two distributions.
\begin{lemma} \label{l:cpleq}
For every $\mu_1,\mu_2\in\Dist(A)$, the following are equivalent:
\begin{itemize}
\item $\mu_1 =\mu_2$;
\item for every $a\in A$, $\pr{x\sim \mu_1}{x=a} = \pr{x\sim \mu_2}{x=a}$;
\item for every $a\in A$, $\coupledsupp {}{\mu_1} {\mu_2} {\post_a}$ where
  $\post_a \eqdef \{ (x_1,x_2) \mid x_1=a \iff x_2=a \}$;
\item $\coupledsupp {}{\mu_1} {\mu_2} {\post_A}$ where
  $\post_A \eqdef \{ (x_1,x_2) \mid x_1= x_2 \}$.
\end{itemize}
\end{lemma}
We note that the third item (existence of liftings for pointwise
equality) is often easier to establish than the last item (existence
of lifting for equality), since one can choose the coupling for each
possible value of $a$, rather than showing a single coupling for all
values of $a$.


\section{Setting}\label{sec:setting}
We will work with a simple probabilistic imperative language. Probabilistic
assignments are of the form $x \rnd g$, which assigns a value sampled
according to the distribution $g$ to the program variable $x$. The
syntax of statements is defined by the grammar:
\begin{align*}
    s &::= \iskip
           \mid \iabort
           \mid \iass x  e
           \mid \irnd x  g 
           \mid s; s
           \mid \icond{e}{s}{s}
           \mid \iwhile{e}{s}
\end{align*}
where $x$, $e$ and $g$ respectively range over (typed) variables in
$\Var$, expressions in $\expr$ and distributions in $\dexpr$. To ensure that the
set of states is countable, we require that there are finitely many variables
$\Var$.  As usual $\expr$ is defined inductively from $\Var$ and a set
$\mathcal{F}$ of simply typed function symbols. In this paper, distributions
used for sampling are either uniform distributions over a finite type $A$, or
the Bernoulli distribution with parameter $p$, which we denote by
$\bernD(p)$. We assume that expressions and statements are typed in the
usual way. 

We assume we are given a set-theoretical interpretation for every type and
operator of the language. We define a state as a type-preserving
mapping from variables to values, and we let $\Mem$ denote the set of
states. The set of states is equipped with the usual functions for
reading and writing a value; we use $m(x)$ to denote the value of $x$
in $m$, and $m\subst{x}{v}$ to denote state update, in this case the
state obtained from $m$ by updating the value of $x$ with $v$.

One can equip $\PMem$ with a monadic structure, using the Dirac
distributions $\dunit{x}$ for the unit and \emph{distribution
  expectation} $\dlet x \mu {M(x)}$ for the bind, where
\begin{equation*}
 \dlet x \mu {M(x)} : x \mapsto \sum_{a} \mu(a) \cdot M(a)(x).
\end{equation*}
The semantics of expressions and distribution expressions is
parametrized by a state $m$, and is defined in the usual way where we
require all distribution expressions to be interpreted as proper
distributions (sub-distributions with weight $1$).
\begin{definition}[Semantics of statements] \strut
  \begin{itemize}
  \item The semantics $\dsem{m}{s}$ of a statement $s$ w.r.t.\ to
    some initial state $m$ is a sub-distribution over states, and is
    defined by the clauses of \cref{fig:semantics}.

  \item The (lifted) semantics $\dsem{\mu}{s}$ of a statement $s$
    w.r.t.\ to some initial sub-distribution $\mu$ over states is a
    sub-distribution over states, and is defined as
    $\dsem{\mu}{s}\triangleq \dlet m \mu {\dsem{m}{s}}$ $\mu \in
    \PMem$.
  \end{itemize}
\end{definition}  
A basic and highly important property of probabilistic programs is termination.
We say that a program $s$ is \emph{lossless} if for every initial memory $m$,
$\wt{\dsem{m}{s}}=1$. By now, there are many sophisticated techniques for
proving losslessness even for languages that allow both probabilistic sampling
and non-determinism (including recent advances
by~\citet{ferrer2015probabilistic,DBLP:journals/corr/ChatterjeeFNH15,DBLP:journals/corr/ChatterjeeNZ16}).
These techniques are capable of showing losslessness for all of our examples (in
some cases with a high degree of automation), so throughout the paper, we assume
that all programs are lossless. This assumption is used in the rules of \Sprhl
and the characterizations of uniformity and independence.

\begin{figure*}
\[\begin{array}{c}
\begin{aligned}
  \dsem{m}{\iskip} &= \dunit{m} &
  \dsem{m}{\iabort} &= \dnull \\
  \dsem{m}{x \asn e} &= \dunit{m\subst{x}{\dsem{m}{e}}} &
  \dsem{m}{x \rnd g} &= \dlet v {\denot{g}{m}} {\dunit{m[x:=v]}}
\end{aligned} \\[1.5em]
\begin{aligned}
  \dsem{m}{s_1; s_2} &= \dlet {\xi} {\dsem{m}{s_1}} {\dsem{\xi}{s_2}} \\
  \dsem{m}{\icond{e}{s_1}{s_2}} &=
    \text{if $\dsem{m}{e}$ then $\dsem{m}{s_1}$ else $\dsem{m}{s_2}$} \\
    \dsem{m}{\iwhile{b}{s}} &= \lim_{n \to \infty}\ \dsem{m}{
    \itfor{(\icondT{b}{s})}{n}; \icondT{b}{\iabort}}
\end{aligned}
\end{array}\]
where $\itfor{s}{n} \triangleq \overbrace{s;\ldots;s}^{n~\mbox{times}}$.
\caption{\label{fig:semantics} Denotational semantics of programs}
\end{figure*}

\subsection{Self-Composition of Programs}
For every program $s$ and $n\in\mathbb{N}$, we let $\selfcomp{s}{n}$
denote the $n$-fold self-composition of $s$, i.e.\, $\selfcomp{s}{n}
\eqdef s_1; \ldots, s_n$, where each $s_\imath$ is a copy of $s$ where
all variables are tagged with a superscript $\imath$. In order to
state the main property of self-composition, we define the
self-composition of a state; given a state $m$, we define its $n$-fold
self-composition $\selfcomp{m}{n}$ as the state from
$\selfcomp{\mathcal{X}}{n}$ to values, where
$\selfcomp{\mathcal{X}}{n} \eqdef \{ x^\imath \mid x\in\mathcal{X},
1\leq \imath \leq n \}$ such that for every $x$ and $\imath$,
$\selfcomp{m}{n}(x^\imath)\eqdef m(x)$.  Given a state $m$ from
$\selfcomp{\mathcal{X}}{n}$, we denote by $m_\imath$ the $\imath$-th
projection of $m$.


\begin{prop}\label{prop:self-comp}
For every program $s$ and state $m$, we have
$$\pr{\dsem{\selfcomp{m}{n}}{\selfcomp{s}{n}}}{\wedge_{1\leq \imath
    \leq n} E^\imath_\imath} =\prod_{1\leq \imath \leq
  n}\pr{\dsem{m}{s}}{E_\imath}$$ where the event $E^\imath$ is defined
by $E^\imath (\selfcomp{m'}{n})\eqdef E(\pi_\imath(m'))$ for every
$\imath$ and $\pi_\imath$ is the projection from a self-composed state
to its $\imath$-th component.
\end{prop}

\subsection{Probabilistic Relational Hoare Logic}\label{sec:prhl}
Probabilistic Relational Hoare Logic (\Sprhl) is a program logic for
reasoning about relational properties of probabilistic programs. Its
judgments are of the form $\Equiv{s_1}{s_2}{\phi}{\psi}$, where $s_1$
and $s_2$ are commands and the pre-condition $\phi$ and the
post-condition $\psi$ are relational assertions, i.e.\  first-order
formulae built over generalized expressions. The latter are similar to
expressions, except that each variable is tagged with $\sidel$ or
$\sider$ to indicate the execution that it belongs to; we call the two
executions \emph{left} and \emph{right}. Generalized expressions are
interpreted w.r.t.\ a pair $(m_1,m_2)$ of states, where the
interpretation of the tagged variables $x\sidel$ and $x\sider$ are
$m_1(x)$ and $m_2(x)$ respectively.  We write $(m_1,m_2)\vDash \phi$
to denote that the interpretation of the assertion $\phi$ w.r.t.\
$(m_1,m_2)$ is valid.
\begin{definition}
A judgment $\Equiv{s_1}{s_2}{\phi}{\psi}$ is valid iff for every
states $m_1$ and $m_2$, $(m_1,m_2)\vDash \phi$ implies
$\coupledsupp {} 
{\dsem{m_1}{s_1}} {\dsem{m_2}{s_2}} {\{ (m'_1,m'_2) \mid (m'_1,m'_2)
  \vDash \psi\}}$.
\end{definition}
\cref{fig:rules} presents the main rules of the logic; see
\citet{BartheGZ09,BGHS17} for the full system. The logic includes
\emph{two-sided} rules, which operate on both programs, and \emph{one-sided}
rules, which operate on a single program (left or right).

The [\textsc{Conseq}] rule is the rule of consequence, and
reflects that validity is preserved by weakening the post-condition
and strengthening the pre-condition.
The [\textsc{Case}] rule allows proving a judgment by case analysis;
specifically, the validity of a judgment with pre-condition $\pre$ can
be established from the validity of two judgments, one where the
pre-condition is strengthened with $\Xi$ and the other where the
pre-condition is strengthened with $\neg \Xi$.

The [\textsc{Struct}] rule allows replacing programs by provably
equivalent programs. The rules for proving program equivalence are
given in \cref{fig:equiv}, and manipulate judgments of the form
$\eqsem{\pre}{c}{c'}$, where $\pre$ is a relational assertion. The first rule
([\textsc{While-Split}]) splits a single loop into two loops (the first running
while $e'$ is true, and the second running for the remaining iterations); this
transformation is useful for selecting different couplings in different program
iterations. The second rule ([\textsc{Swap}]) reorders two instructions, as long
as they modify disjoint variables. This allows us to couple sampling
instructions that may come from two different parts of the two programs.

Moving on to the two-sided rules,
the [\textsc{Seq}] rule for sequential composition simply reflects the
compositional property of couplings.
The [\textsc{Assg}] rule is standard.
The [\textsc{Rand}] rule informally takes a coupling between the two
distributions used for sampling in the left and right programs, and requires
that every element in the support of the coupling validates the post-condition.
The rule is parametrized by a bijective function $f$ from the domain of the
first distribution to the domain of the second distribution.  This
bijection gives us the freedom to specify the relation between the two samples
when we couple them.
The [\textsc{Cond}] rule states that two \emph{synchronized} \kif\
statements can related if their respective branches are also related.
The [\textsc{While}] rule is the standard while rule adapted to
pRHL. Note that we require the guard of the two commands to be
equal—so in particular the two loops must make the same number of
iterations—and $\pre$ plays the role of the while loop invariant as
usual.

The one-sided rules \textsc{Assg-L}, \textsc{Rand-L}, \textsc{Cond-L} and
\textsc{While-L} are similar two their two-sided variant, but only operate on
the left program. The full system includes mirrored versions of each one-sided
rule, for reasoning about the right program.


\begin{figure*}
\begin{mathpar}
\xinferrule*[left=\textsc{Conseq}]
 { \Equiv{s_1}{s_2}{\pre}{\post} \\\\
   \pre'\implies \pre \\
   \post \implies \post'}
 { \Equiv{s_1}{s_2}{\pre'}{\post'}}

\xinferrule*[left=\textsc{Struct}]
 {\Equiv{s_1}{s_2}{\pre}{\post} \\\\
   \eqsem{\pre}{s_1}{s_1'} \\ \eqsem{\pre}{s_2}{s_2'}}
 {\Equiv{s_1'}{s_2'}{\pre}{\post}}

\xinferrule*[left=\textsc{Case}]
 {\Equiv{s_1}{s_2}{\pre\land \Xi}{\post} \\\\
  \Equiv{s_1}{s_2}{\pre\land \neg \Xi}{\post}}
 { \Equiv{s_1}{s_2}{\pre}{\post}}

\xinferrule*[left=\textsc{Seq}]
 {\Equiv{s_1}{s_2}{\pre}{\Xi} \\\\
  \Equiv{s'_1}{s'_2}{\Xi}{\post}}
 {\Equiv{s_1;s'_1}{s_2;s'_2}{\pre}{\post}}

\xinferrule*[left=\textsc{Assg}]
  {\pre \eqdef \post[e_1\sidel/x_1\sidel,e_2\sider/x_2\sider]}
  {\Equiv
      {\iass{x_1}{e_1}}{\iass{x_2}{e_2}}
      {\pre}{\post}}

\xinferrule*[left=\textsc{Rand}]
  {\coupledsupp{f}{{g_1}}{{g_2}}{} \\\\
   \pre \eqdef \forall v .\, \post[v/x_1\sidel,f(v)/x_2\sider]}
  {\Equiv
      {\irnd{x_1}{g_1}}{\irnd{x_2}{g_2}}
      {\pre}{\post}}

\xinferrule*[left=\textsc{Cond}]
  {\pre \implies e_1 = e_2 \\\\
   \Equiv{s_1}{s_2}{\pre \land e_1}{\post}{s} \\
   \Equiv{s_1'}{s_2'}{\pre \land \neg e_1}{\post}{s'}}
 {\Equiv
     {\icond{e_1}{s_1}{s_1'}}{\icond{e_2}{s_2}{s_2'}}
     {\pre}{\post}}

\xinferrule*[left=\textsc{While}]
              {
                \Equiv{s_1}{s_2}{\post \land e_1\sidel \land e_2\sider}{
                  \post\land e_1\sidel=e_2\sider}}
  {\Equiv{\iwhile{e_1}{s_1}}{\iwhile{e_2}{s_2}}
      {\post\land e_1\sidel=e_2\sider}
      {\post \land \neg e_1\sidel \land \neg e_2\sider}}

  \xinferrule*[left=\textsc{Assg-L}]
  {\pre \eqdef \post[e_1\sidel/x_1\sidel]}
  {\Equiv
      {\iass{x_1}{e_1}}{\iskip}
      {\pre}{\post}
      }

\xinferrule*[left=\textsc{Rand-L}]
  {\pre \eqdef \forall v_1\in\supp(g_1), \post [v_1/x_1\sidel]}
  {\Equiv
      {\irnd{x_1}{g_1}}{\iskip}
      {\pre}
      {\post}}
  
\xinferrule*[left=\textsc{Cond-L}]
  {\Equiv{s_1}{s_2}{\pre \land e_1\sidel}{\post} \\\\
   \Equiv{s_1'}{s_2}{\pre \land \neg e_1\sidel}{\post}}
  {\Equiv
     {\icond{e_1}{s_1}{s_1'}}{s_2}
     {\pre}{\post}}
 
\xinferrule*[left=\textsc{While-L}]
  {\Equiv{s_1}{\iskip}{\post \land e_1\sidel}{\post}}
  {\Equiv{\iwhile{e_1}{s_1}}{\iskip}
     {\post}{\post\land\neg e_1\sidel}}
\end{mathpar}

\caption{\label{fig:rules} Proof rules (selection)}
\end{figure*}

\begin{figure}
\begin{mathpar}
\xinferrule*[left=\textsc{While-Split}]
{~}{\eqsem{\pre}{\iwhile{e}{s}}{\iwhile{e\land e'}{s};\iwhile{e}{s}}}

\xinferrule*[left=\textsc{Swap}]{\mathsf{var}(s_1) \cap \mathsf{var}(s_2)=\emptyset}{
\eqsem{\pre}{s_1;s_2}{s_2;s_1}}

\end{mathpar}
\caption{Equivalence rules (selection)}\label{fig:equiv}
\end{figure}

Throughout the paper, we often assert that the left and the right
copies of a state are equal. This is captured by the relational
assertion $\EqMem\eqdef \bigwedge_{x\in \mathcal{X}} x\sidel=
x\sider$.  We also often assert cross-equality on $n$-fold composition
of states $\Eqmem{p}{q} \eqdef \bigwedge_{x\in \mathcal{X},1\leq
  \imath \leq p , 1 \leq \jmath \leq q} x^\imath\sidel
=x^\jmath\sider$.

\section{Uniformity}\label{sec:unif}
Reasoning about probabilistic programs often requires establishing
that a set of program variables (each ranging over a finite type) is
uniformly distributed:
\begin{definition}
A set $X=\{ x_1,\ldots, x_n\}$ of program variables of finite types
$A_1,\ldots, A_n$ is \emph{uniformly distributed} in a
distribution $\mu\in\Dist(\Mem)$ iff for every $ (a_1,\ldots, a_n) \in
A_1 \times \ldots \times A_n$:
$$\pr{\mu}{\bigwedge_{1\leq i \leq n} x_i = a_i} =\prod_{1 \leq i \leq n}
\frac{1}{|A_i|}$$
\end{definition}
Note that the definition of uniformity (and as we will see in later sections,
the definition of independence) naturally extends to sets of expressions, and so
do our characterizations.

\subsection{Characterization}

The following proposition characterizes uniformity in terms of couplings.
\begin{prop}[Uniformity by coupling]\label{prop:uni-coupling}
Let $X=\{ x_1, \ldots, x_n\}$ be a set of variables of respective
finite types $A_1,\ldots, A_n$. For every program $s$, the
following are equivalent:
\begin{enumerate}
\item for every state $m$, $X$ is uniformly distributed in
  $\dsem{m}{s}$;
      
\item 
  for every two tuples $(a_1,\ldots,a_n),(a'_1,\ldots,a'_n)\in
  A_1\times \cdots \times A_n$, we have
  \[
    \Equiv{s}{s}{\EqMem}
    {\left(\bigwedge_{1\leq i \leq n} x_i \sidel = a_i\right) \iff
     \left(\bigwedge_{1\leq i \leq n} x_i \sider = a'_i \right)} .
  \]
\end{enumerate}
\end{prop}

\begin{proof}
\begin{description}[leftmargin=0cm]
\item[{[$1. \Rightarrow 2.$]}] Let $m$ be a memory and assume that $X$
  is uniformly distributed in $\dsem{m}{s}$. Let
  $(a_1, \ldots, a_n), (a'_1, \ldots, a'_n) \in A_1 \times \cdots
  \times A_n$.
  We denote by $f : \Mem \to \Mem$ the bijection defined by
  \[ \left\{ \begin{aligned}
    f (m) &= m[x_i \leftarrow a'_i]_{1 \leq i \leq n}
          & \text{if $\forall i .\, m[x_i] = a_i$} \\
    f (m) &= m[x_i \leftarrow a_i]_{1 \leq i \leq n}
          & \text{if $\forall i .\, m[x_i] = a'_i$} \\
    f (m) &= m &\text{otherwise.}
  \end{aligned} \right. \]
  Let $\eta \in \Dist(\Mem \times \Mem)$ be the distribution defined by
  $\eta(m_1, m_2) = \dsem{m}{s}(m_1)$ if $m_2 = f(m_1)$, and
  $\eta(m_1, m_2) = 0$ otherwise. We prove that $\eta$ is a
  $\Psi$-coupling for $\dsem{m}{s}$, where
  \[ \psi \eqdef
     \left(\textstyle\bigwedge_{1\leq i \leq n} x_i \sidel = a_i\right) \iff
     \left(\textstyle\bigwedge_{1\leq i \leq n} x_i \sider = a'_i \right) .\]
  Regarding the marginals, we have:
  \begin{align*}
    \proj_1(\eta)(m_1)
      &= \sum_{m_2} \eta(m_1, m_2)
       = \eta(m_1, f(m_1)) = \dsem{m}{s}(m_1) \\
    \proj_2(\eta)(m_2)
      &= \sum_{m_1} \eta(m_1, m_2)
       = \eta(f^{-1}(m_2), m_2) = \dsem{m}{s}(f^{-1}(m_2)) \\
      &= \dsem{m}{s}(f(m_2)) = \dsem{m}{s}(m_2),
  \end{align*}
  the last equality being a consequence of $X$ being uniformly
  distributed in $\dsem{m}{s}$. Moreover, for
  $(m_1, m_2) \in \supp(\eta)$, we have $m_2 = f(m_1)$. Thus,
  $m_1[x_i] = a_i$ iff $m_2[x_i] = a'_i$, and $m_1, m_2 \models \Psi$.

\item[{[$2. \Rightarrow 1.$]}]
  Let
  $(a_1, \ldots, a_n), (a'_1, \ldots, a'_n) \in A_1 \times \cdots \times A_n$
  and assume that $\Equiv{s}{s}{\EqMem}{\Psi}$, where $\Psi$ is
  defined as in the previous case. Since $m, m \models \EqMem$, by
  \cref{l:cplfd} we have:
  \[ \pr{\dsem{m}{s}}{\textstyle\bigwedge_{1 \leq i \leq n} x_i = a_i} =
     \pr{\dsem{m}{s}}{\textstyle\bigwedge_{1 \leq i \leq n} x_i = a'_i} , \]
  showing that $X$ is uniform in $\dsem{m}{s}$.
  \qedhere
\end{description}
\end{proof}

By expressing uniformity as a coupling property, we can use \Sprhl\ to
prove uniformity. To demonstrate the technique, we consider classical
examples from the theory of randomized algorithms.

\subsection{Simulating a Fair Coin}

\begin{wrapfigure}{l}{5.3cm}
  $$\boxed{\begin{array}{l}
  \iass{x}{\pfalse}; \\
  \iass{y}{\pfalse}; \\
  \iwhile{x = y}{} \\
  \quad x \rnd {\bernD(p)}; \\
  \quad y \rnd {\bernD(p)}; \\
  \end{array}}
  $$
  \caption{Bernoulli uniformizer}\label{fig:bern}
\end{wrapfigure}
This example considers a process for simulating a fair coin using a
biased coin. The idea is simple: 1) toss the coin twice; 2) if the two
outcomes differ, return the value of the first coin; 3) if the two
outcomes match, repeat from step 1.  The algorithm does not require
the bias of the coin to be known, as long as it is some constant bias
and there is positive probability of returning $0$ and $1$.  This
process can be modelled by the program $s$ from \cref{fig:bern},
where $0 < p < 1$ is a real parameter modeling the probability of the
biased coin to return 0 (tail). Our goal is to establish the trivial
judgment $\hoare{\True}{s}{\True}$ and the following \Sprhl judgment:
$$
\Equiv{s}{s}{\True}{x\sidel \iff \neg x\sider}
$$
By the fundamental lemma of coupling, this implies that
$\pr{\dsem{m}{s}}{x=\ptrue}= \pr{\dsem{m}{s}}{x=\pfalse}$,
and hence that $x$ is uniformly distributed upon termination.
The proof proceeds by establishing the following invariant:
$$x\sider = \textrm{if } x\sidel = y\sidel \textrm{ then } y\sider \textrm{ 
else }\neg x\sidel$$
Validity of the invariant entails that the desired postcondition holds
when the program exits, as the invariant and the negation of the loop guard both
hold. The invariant holds when entering the loop, so we only need to
prove that it is preserved by the loop body. The proof proceeds as
follows: first, we swap the two random assignments on the right,
leading to the judgment:
$$ \Equiv{(x \rnd \bernD(p); \quad y \rnd \bernD(p))} {(y \rnd \bernD(p); 
\quad x \rnd \bernD(p))}{\phi'}{ \phi}
$$
where $\phi$ denotes the loop invariant and $\phi'$ denotes its
strengthening by the loop guard---we do not need the precondition,
since the values are freshly sampled in the body. Next, we apply the
\rname{Rand} rule twice, with the identity bijection. The required pre-condition
\[
  \forall v_1, v_2,\,
  v_2 = (\text{if $v_1 = v_2$ then $v_1$ else $\neg v_1$})
\]
is clearly true.

\subsection{Cyclic Random Walk}

\ifbw
\colorlet{ccred}{gray}
\else
\colorlet{ccred}{red}
\fi

\begin{figure}
\centering
\begin{tikzpicture}[baseline]
\tikzstyle{every node}=[font=\scriptsize]
\def\radius{1cm}

\draw[thick] ([shift=(0:\radius)]0,0) arc (0:135:\radius);
\draw[thick] ([shift=(-45:\radius)]0,0) arc (-45:-135:\radius);

\fill[ccred!05!white] (0,0) --
  ([shift=(0:\radius)]0,0) arc (0:-45:\radius) -- cycle;

\draw[ccred,densely dashed,thick] ([shift=(0:\radius)]0,0) arc (0:-45:\radius);

\draw[thick,loosely dotted]
  ([shift=(135:\radius)]0,0) arc (135:225:\radius);

\fill (0,0) circle[radius=2pt];

\fill (0,0) ++ (0:\radius) circle[radius=2pt]
            ++ (0:0.7em) node {$0$};

\fill (0,0) ++ (45:\radius) circle[radius=2pt]
            ++ (45:1.0em) node {$n{-}1$};

\fill (0,0) ++ (90:\radius) circle[radius=2pt];

\fill (0,0) ++ (135:\radius) circle[radius=2pt]
            ++ (135:0.7em) node {$f$};

\fill (0,0) ++ (-45:\radius) circle[radius=2pt]
            ++ (-45:0.7em) node {$1$};

\fill (0,0) ++ (-90:\radius) circle[radius=2pt]
            ++ (-90:0.7em) node {$2$};

\fill (0,0) ++ (-135:\radius) circle[radius=2pt]
            ++ (-135:0.7em) node {$l$};

\draw[->,>=stealth]
  ([shift=(0:0.35*\radius)]0,0) arc (0:-45:0.35*\radius);

\draw[thick,dotted,gray,->,>=stealth] (0,0) -- ++ (0:.7*\radius);

\draw[thick,->,>=stealth] (0,0) -- ++ (-45:.7*\radius)
  node[left] {$c~$};
\end{tikzpicture}
\hspace*{1.5cm}
$\boxed{\begin{array}[c]{l}
   \iass{d}{0};~ \iass{c}{0};~ \iass{f}{0};~ \iass{l}{0}; \\
   \iwhile{l + 1 \leq f}{} \\
   \quad \irnd{d}{\mathcal{U}_{\{-1 ,1\}}}; \\
   \quad \icondT{c = l \land d = 1}{ \iass{l}{l+1};} \\
   \quad \icondT{c = f \land d = -1}{ \iass{f}{f-1};} \\
   \quad \iass{c}{c + d}; \\
   \iass{ret}{(l, l+1)}
 \end{array}}$
\caption{\label{fig:walk} Cyclic random walk}
\end{figure}
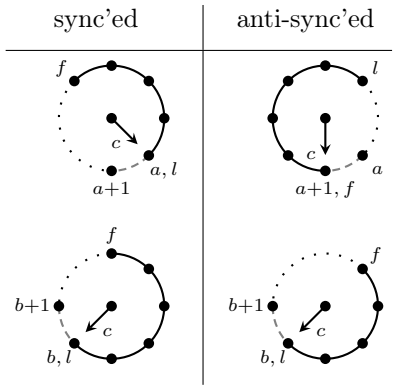

Consider a random walk over a cyclic path composed of $n$ nodes
labeled $0,1,\ldots, n - 1$: starting from position $0$, at each step,
we flip a fair coin over $\{-1, 1\}$ and update the position
accordingly to the result of the coin flip. To take into account that
we are on a cyclic structure, all arithmetical operations are in the
cyclic ring $\Zn{n}$---i.e. are performed modulo $n$.  At each
iteration, when moving between two contiguous positions over the circle,
we consider that the random walk visited the arc between the two
nodes.  We want to show that the last visited arc is uniformly
distributed.
\cref{fig:walk}~(left) gives a graphical representation of the random
walk, where $c$ is the random walk position and the dashed arc is the
last visited arc when c moved from $0$ to $1$.

\begin{wraptable}{l}{5.5cm}
\def\radius{0.7cm}
\begin{tabular}{@{}c|c@{}}
sync'ed & anti-sync'ed \\[.5em]
\hline
\begin{tikzpicture}
\tikzstyle{every node}=[font=\scriptsize]
\draw[thick]
  ([shift=(-45:\radius)]0,0) arc (-45:135:\radius);

\draw[thick,loosely dotted]
  ([shift=(135:\radius)]0,0) arc (135:270:\radius);

\draw[ccred,thick,densely dashed] ([shift=(-45:\radius)]0,0) arc (-45:-90:\radius);

\fill (0,0) circle[radius=2pt];

\fill (0,0) ++ (-45:\radius) circle[radius=2pt]
            ++ (-45:0.8em) node {$a, l$};

\fill (0,0) ++ (  0:\radius) circle[radius=2pt];
\fill (0,0) ++ ( 45:\radius) circle[radius=2pt];
\fill (0,0) ++ ( 90:\radius) circle[radius=2pt];

\fill (0,0) ++ (135:\radius) circle[radius=2pt]
            ++ (135:0.7em) node {$f$};

\fill (0,0) ++ (-90:\radius) circle[radius=2pt]
            ++ (-90:0.7em) node {$a{+}1$};

\draw[thick,->,>=stealth] (0,0) -- ++ (-45:.7*\radius)
  node[left] {$c~$};

\draw[thick]
  ([shift=(-135:\radius)]0,-2.5) arc (-135:90:\radius);

\draw[thick,loosely dotted]
  ([shift=(90:\radius)]0,-2.5) arc (90:180:\radius);

\draw[ccred,thick,densely dashed] ([shift=(180:\radius)]0,-2.5)
  arc (180:225:\radius);

\fill (0,-2.5) circle[radius=2pt];

\fill (0,-2.5) ++ (-135:\radius) circle[radius=2pt]
            ++ (-135:0.8em) node {$b, l$};

\fill (0,-2.5) ++ (-90:\radius) circle[radius=2pt];
\fill (0,-2.5) ++ (-45:\radius) circle[radius=2pt];
\fill (0,-2.5) ++ (  0:\radius) circle[radius=2pt];
\fill (0,-2.5) ++ ( 45:\radius) circle[radius=2pt];

\fill (0,-2.5) ++ (90:\radius) circle[radius=2pt]
               ++ (90:0.7em) node {$f$};

\fill (0,-2.5) ++ (-180:\radius) circle[radius=2pt]
               ++ (-180:1.0em) node {$b{+}1$};

\draw[thick,->,>=stealth] (0,-2.5) -- ++ (-135:.7*\radius)
  node[right] {$~c$};
\end{tikzpicture} &
\begin{tikzpicture}
\tikzstyle{every node}=[font=\scriptsize]
\draw[thick]
  ([shift=(45:\radius)]0,0) arc (45:270:\radius);

\draw[thick,loosely dotted]
  ([shift=(45:\radius)]0,0) arc (45:-45:\radius);

\draw[ccred,thick,densely dashed] ([shift=(-45:\radius)]0,0) arc (-45:-90:\radius);

\fill (0,0) circle[radius=2pt];

\fill (0,0) ++ (-45:\radius) circle[radius=2pt]
            ++ (-45:0.8em) node {$a$};

\fill (0,0) ++ ( 45:\radius) circle[radius=2pt]
            ++ ( 45:0.7em) node {$l$};

\fill (0,0) ++ ( 90:\radius) circle[radius=2pt];
\fill (0,0) ++ (135:\radius) circle[radius=2pt];
\fill (0,0) ++ (180:\radius) circle[radius=2pt];
\fill (0,0) ++ (225:\radius) circle[radius=2pt];

\fill (0,0) ++ (-90:\radius) circle[radius=2pt]
            ++ (-90:0.7em) node {$a{+}1,f$};

\draw[thick,->,>=stealth] (0,0) -- ++ (-90:.7*\radius)
  node[left] {$c$};

\draw[thick]
  ([shift=(-135:\radius)]0,-2.5) arc (-135:45:\radius);

\draw[thick,loosely dotted]
  ([shift=(45:\radius)]0,-2.5) arc (45:180:\radius);

\draw[ccred,thick,densely dashed] ([shift=(180:\radius)]0,-2.5)
  arc (180:225:\radius);

\fill (0,-2.5) circle[radius=2pt];

\fill (0,-2.5) ++ (-135:\radius) circle[radius=2pt]
            ++ (-135:0.8em) node {$b, l$};

\fill (0,-2.5) ++ (-90:\radius) circle[radius=2pt];
\fill (0,-2.5) ++ (-45:\radius) circle[radius=2pt];
\fill (0,-2.5) ++ (  0:\radius) circle[radius=2pt];
\fill (0,-2.5) ++ ( 45:\radius) circle[radius=2pt]
               ++ ( 45:0.7em) node {$f$};

\fill (0,-2.5) ++ (-180:\radius) circle[radius=2pt]
               ++ (-180:1.0em) node {$b{+}1$};

\draw[thick,->,>=stealth] (0,-2.5) -- ++ (-135:.7*\radius)
  node[right] {$~c$};
\end{tikzpicture}
\end{tabular}
\end{wraptable}

This process can be seen as a simple version of an algorithm that samples a
uniformly random spanning tree on a graph---when the graph is a cycle, a
spanning tree visits all but one of the edges. While \citet{Broder89}
analyzes the general problem, we can verify uniformity for the cyclic
random walk with couplings.

The proof proceeds as follows. We imagine executing two random walks, from the
same initial position. The goal is to couple the walks so that $(a, a+1)$ is the
last arc in the first walk if and only if $(b, b+1)$ is the last arc in the
second walk. If we can show this property for all $a, b$, then this coupling
argument shows that any two arcs have the same probability of being the last
arc, hence the last arc must be uniformly distributed.

\begin{wraptable}{r}{5.7cm}
\def\radius{0.7cm}
\begin{tabular}{@{}c|c@{}}
case (i) & case (ii) \\[.5em]
\hline
\begin{tikzpicture}
\tikzstyle{every node}=[font=\scriptsize]
\draw[thick]
  ([shift=(-45:\radius)]0,0) arc (-45:270:\radius);

\draw[thick,loosely dotted]
  ([shift=(-90:\radius)]0,0) arc (-90:-45:\radius);

\fill (0,0) circle[radius=2pt];

\fill (0,0) ++ (-45:\radius) circle[radius=2pt]
            ++ (-45:0.8em) node {$a,l$};

\fill (0,0) ++ (  0:\radius) circle[radius=2pt];
\fill (0,0) ++ ( 45:\radius) circle[radius=2pt];
\fill (0,0) ++ ( 90:\radius) circle[radius=2pt];
\fill (0,0) ++ (135:\radius) circle[radius=2pt];
\fill (0,0) ++ (180:\radius) circle[radius=2pt];
\fill (0,0) ++ (225:\radius) circle[radius=2pt];

\fill (0,0) ++ (-90:\radius) circle[radius=2pt]
            ++ (-90:1.0em) node {$a{+}1, f$};

\draw[thick,->,>=stealth] (0,0) -- ++ (-90:.7*\radius);

\draw[thick]
  ([shift=(-135:\radius)]0,-2.5) arc (-135:180:\radius);

\draw[thick,loosely dotted]
  ([shift=(180:\radius)]0,-2.5) arc (180:225:\radius);

\fill (0,-2.5) circle[radius=2pt];

\fill (0,-2.5) ++ (-135:\radius) circle[radius=2pt]
            ++ (-135:0.8em) node {$b, l$};

\fill (0,-2.5) ++ (-90:\radius) circle[radius=2pt];
\fill (0,-2.5) ++ (-45:\radius) circle[radius=2pt];
\fill (0,-2.5) ++ (  0:\radius) circle[radius=2pt];
\fill (0,-2.5) ++ ( 45:\radius) circle[radius=2pt];
\fill (0,-2.5) ++ ( 90:\radius) circle[radius=2pt];
\fill (0,-2.5) ++ (135:\radius) circle[radius=2pt];

\fill (0,-2.5) ++ (-180:\radius) circle[radius=2pt]
               ++ (-180:1.0em) node
               {$\begin{array}{c}b{+}1 \\ f\end{array}$};

\draw[thick,->,>=stealth] (0,-2.5) -- ++ (180:.7*\radius);
\end{tikzpicture} &
\begin{tikzpicture}
\tikzstyle{every node}=[font=\scriptsize]
\draw[thick]
  ([shift=(45:\radius)]0,0) arc (45:315:\radius);

\draw[thick,loosely dotted]
  ([shift=(45:\radius)]0,0) arc (45:-45:\radius);

\fill (0,0) circle[radius=2pt];

\fill (0,0) ++ (-45:\radius) circle[radius=2pt]
            ++ (-45:0.8em) node {$a,f$};

\fill (0,0) ++ (  0:\radius) circle[radius=2pt];

\fill (0,0) ++ ( 45:\radius) circle[radius=2pt]
            ++ ( 45:0.7em) node {$l$};

\fill (0,0) ++ ( 90:\radius) circle[radius=2pt];
\fill (0,0) ++ (135:\radius) circle[radius=2pt];
\fill (0,0) ++ (180:\radius) circle[radius=2pt];
\fill (0,0) ++ (225:\radius) circle[radius=2pt];

\fill (0,0) ++ (-90:\radius) circle[radius=2pt]
            ++ (-90:0.7em) node {$a{+}1$};

\draw[thick,->,>=stealth] (0,0) -- ++ (-45:.7*\radius);

\draw[thick]
  ([shift=(-180:\radius)]0,-2.5) arc (-180:90:\radius);

\draw[thick,loosely dotted]
  ([shift=(90:\radius)]0,-2.5) arc (90:180:\radius);

\fill (0,-2.5) circle[radius=2pt];

\fill (0,-2.5) ++ (-135:\radius) circle[radius=2pt]
            ++ (-135:0.8em) node {$b$};

\fill (0,-2.5) ++ (-90:\radius) circle[radius=2pt];
\fill (0,-2.5) ++ (-45:\radius) circle[radius=2pt];
\fill (0,-2.5) ++ (  0:\radius) circle[radius=2pt];
\fill (0,-2.5) ++ ( 45:\radius) circle[radius=2pt];
\fill (0,-2.5) ++ ( 90:\radius) circle[radius=2pt]
               ++ ( 90:0.7em) node {$f$};

\fill (0,-2.5) ++ (135:\radius) circle[radius=2pt];

\fill (0,-2.5) ++ (-180:\radius) circle[radius=2pt]
               ++ (-180:1.0em) node
               {$\begin{array}{c}b{+}1 \\ l\end{array}$};

\draw[thick,->,>=stealth] (0,-2.5) -- ++ (180:.7*\radius);
\end{tikzpicture}
\end{tabular}
\end{wraptable}

To describe the coupling informally, we first execute asynchronously
the two random walks until they eventually synchronize respectively on
the arcs $(a, a+1)$ and $(b, b+1)$. At that point, we are in one of
the following cases: either the random walks synchronize on the same
side of the arcs $(a, a+1)$ and $(b, b+1)$, or they synchronize on
opposite sides. (These cases are depicted on the left diagrams above,
where the arc we want to synchronize on is dashed.) From that point,
we execute the two processes resp.\ in lock-step (if they synchronized
on the same side) or anti-lock-step (if they did not).

At some point, both processes will visit the other side of the arcs
$(a, a+1)$ and $(b,b+1)$, and since they execute in (anti)-lock-step,
these events will occur synchronously.
At that point, either the processes finished their walk and they
resp.\ return the arcs $(a, a+1)$ and $(b, b+1)$ as their result
(case~(i) of the right diagram above), or they have other nodes to
visit and so they \emph{will not} resp.\ return the arcs $(a, a+1)$ and
$(b, b+1)$ (case~(ii) of the same diagram).

We now detail the formal proof.  Consider the program of
\cref{fig:walk}, where all arithmetical operations are done
modulo $n$. This algorithm instruments the random walk with two points
$f$ and $l$ representing the range $[f,l]$ (using clockwise ordering)
of all the points that have been visited by the walk. When all nodes
of the cycle have been visited (i.e. when $l+1 = f$), the arc between
$l$ and $l+1$ is the only arc that has not been visited by the walk.
Let $s$ be the program of \cref{fig:walk} and $s'$ the loop body
of the single loop of $s$. We want to show that the final arc
$ret$---the only arc that has not been marked---is uniformly
distributed among all arcs. This follows by the judgment:
\[
\forall a, b \in \Zn{n} ,\,
  \Equiv{s}{s}{\top}{ret\sidel=(a, a+1) \iff ret\sider=(b, b+1)} .
\]

First, we make use of the loop splitting equivalence rule (\rname{While-Split}) to
transform the main loop into three pieces. In the left program:
\[
  \begin{array}{l}
    \iwhile{(|v| < n \land a, a+1 \notin [f,l])}{s'};  \\
    \iwhile{(|v| < n \land \neg(a, a+1 \in [f,l]))}{s'};  \\
    \iwhile{(|v| < n)}{s'}
  \end{array}
\]
where $[f,l]$ represent the range $f, f+1, \ldots, l$. We
use a similar transformation on the right program, with $b$ in place
of $a$.  To carry out the proof, we first use the one-sided loop rules
(\rname{While-L} and the corresponding version \rname{While-R}) on the
first loops of the left and right programs. This part of the proof
correspond to the walks synchronization as described above.  By a
straightforward loop invariant, we can show that
\[
  \Phi \land P(a) \sidel \land P(b)\sider
\]
holds after the first loops, where
$P(x) \eqdef (x \in [f,l] \oplus (x + 1) \in [f, l])$
and
$\Phi \eqdef \forall i \in \{1, 2\} .\, (c\in [f,l])\langle i\rangle$
indicates that the current positions ($c\sidel,c\sider$)
are contained in the range of visited arcs.  Next, we show that after
the two second loops the following relational invariant is satisfied:
\[
  (a, a+1 \in [f,l])\sidel \land
  (b, b+1 \in [f,l])\sider \land
  (l\sidel=a \iff l\sider=b)
\]
After the second loop, there are two cases for the third loop. If
$l\sidel=a$, we have $l\sider=b$, $f\sidel=a+1$ (since $a+1$ is
visited in $\sidel$) and $f\sider=b+1$. In this case, which corresponds
to the case~(i) of the last diagram, the third loops both exit
immediately and the random walks resp.\ return the arcs $(a, a+1)$ and
$(b, b+1)$.  Otherwise, we have $l\sidel \neq a$ and $l\sider \neq b$,
and we can show, using the rules \rname{While-L} and \rname{While-R},
that $l\sidel$ (resp.\ $l\sider$) will never be set to $a$
(resp.\ b). In this case, which corresponds to the case~(ii) of the last
diagram, we can show that the walks resp.\ return arcs distinct from
$(a, a+1)$ and $(b, b+1)$.

We now focus on the second loops, relating them with the two-sided
variant of the \rname{While} rule. The particular coupling we choose
will depend on the current positions in the two sides at the start
of the second loops.  If $a \in [f,l]\sidel$ and $b \in [f,l]\sider$
then we have $ l\sidel = a$ and $l\sider = b$ (since
$a+1 \notin[f,l]\sidel$) and we couple the walks to make identical
moves. In that case, the key part of the loop invariant is:
\[ \bigwedge \left\{ \begin{gathered}
  \Phi \land (a \in [f,l])\sidel \land (b \in [f,l])\sider  \\
  c\sidel - a = c\sider - b \\
  l\sidel=a \iff  l\sider = b 
\end{gathered} \right\} \]
The first line enforces some structural invariant and the second line
enforces that both walks make identical moves relative to $a$ and
$b$. The main difficulty is to show that both loops are
synchronized. Note that there are two reasons the loop may exit. If
$l$ has been incremented, then the increment will be done on both
side. Otherwise, if $f$ has been decremented to $a+1$, then we have
$c\sidel = f\sidel = a+2$, so $c\sidel - a = c\sider - b = 2$ and
$c\sider = b + 2$ and the right loop will also decrement $f$ to $b+1$.
The case $a+1 \in [f,l]\sidel$ and $b +1\in [f,l]\sider$ is very
similar, by reversing the roles of $f$ and $l$.
The remaining two cases, $a+1 \in [f,l]\sidel$ and $b\in [f,l]\sider$
or $a \in [f,l]\sidel$ and $b+1\in [f,l]\sider$ is similar except that
we force the walks to be execute in anti-lock-step.
Using the rule \rname{Case}, we put together these four cases and we
conclude by application of the rule for sequence.

\subsection{Ballot Theorem}\label{sec:discussion}

So far, we have shown how couplings can be used to prove that a set of program
variables is uniformly distributed. Couplings can also be used for showing that
two events have the same probability, such as in the following example.

\begin{example}[Ballot Theorem]
Assume that voters must choose between two candidates $A$ and $B$.
The outcome of the vote is $n_A$ votes for $A$ and $n_B$ votes for
$B$, with $n_A > n_B$. Assuming that the order in which the votes are
cast is uniformly random, the probability that $A$ is always strictly
ahead in partial counts is $\nicefrac{(n_A - n_B)}{(n_A + n_B)}$.
\end{example}

The process can be formalized by the program from
\cref{fig:ballot}.  Here we use the list $l$ to store
intermediate results.  Using $l_i$ to denote the $i$-th element of the
list $l$, the Ballot Theorem is captured by the statement:
\[ \forall n_A, n_B. n_A > n_B \implies
     \pr{\dsem{m}{s}}{\left. \bigwedge_{1\leq i \leq n} l_i \neq 0
       ~\right|~
     x_A =n_A \land x_B= n_B}
   = \frac{n_A - n_B }{n_A + n_B} .\]

\begin{wrapfigure}{l}{5.3cm}
\vspace{-4ex}
$$\boxed{\begin{array}{l}
  \iass r 0; \iass {x_A} 0; \iass {x_B} 0; \iass l \epsilon; \\
  \iwhile{|l| \leq n}{} \\
  \quad \irnd r \{ A, B \}; \\
  \quad \icondT{r=A}{}; \\
  \quad\quad \iass {x_A} {x_A+1}; \\
  \quad\kelse \\
  \quad\quad {\iass {x_B} {x_B+1}} ; \\
  \quad \iass l {l::(x_A-x_B)}
\end{array}}$$
\caption{Ballot theorem}\label{fig:ballot}
\end{wrapfigure}

There exist many proofs of the Ballot Theorem; we formalize a proof that is
sometimes called Andre's reflection principle. The crux of the method is a
coupling proof of the following fact: ``bad'' sequences starting with
a vote to the loser are equi-probable with ``bad'' sequences starting
with a vote to the winner, where a sequence of votes is ``bad'' if
there is a tie at some point in the partial counts. Let
$\phi \eqdef ( \bigvee_{1\leq i \leq n} l_i = 0)$ and
$\psi\eqdef x_A = n_A \land x_B = n_B$.
The above facts are captured by the \Sprhl judgment (universally
quantified over $n_A$ and $n_B$ such that $n_A>n_B$):
$\Equiv{s}{s}{\top}{\xi}$ where
\[ \xi \eqdef
     (l_1 \cdot l_n > 0 \land \phi \land \psi)\sidel
      \iff (l_1 \cdot l_n < 0 \land \phi \land \psi)\sider .
\]
It follows from the properties of coupling that for every $n_A$ and
$n_B$ such that $n_A>n_B$,
$\pr{\dsem{m}{s}}{l_1 \cdot l_n > 0 \land \phi \land l_n = k}=
\pr{\dsem{m}{s}}{l_1 \cdot l_n < 0  \land  \phi \land  l_n = k}$.
In terms of conditional probabilities, we have
$\pr{\dsem{m}{s}}{l_1 \cdot l_n > 0 \land \phi \mid \psi}=
\pr{\dsem{m}{s}}{l_1 \cdot l_n < 0  \land  \phi \mid \psi}$.
Now observe that any sequence that starts with a vote to $B$ (i.e.\
the loser) is necessarily bad. Therefore,
$\pr{\dsem{m}{s}}{l_1 \cdot l_n < 0  \land  \phi \mid \psi} =
 \pr{\dsem{m}{s}}{ l_1 \cdot l_n < 0 \mid \psi}$.
By the above and elementary properties of conditional independence:
\begin{align*}
  \pr{\dsem{m}{s}}{\phi \mid \psi} &=
  \pr{\dsem{m}{s}}{l_1 \cdot l_n > 0 \land \phi \mid \psi}  +
  \pr{\dsem{m}{s}}{l_1 \cdot l_n < 0 \land \phi \mid \psi} \\
    & = 2 \cdot \pr{\dsem{m}{s}}{l_1 \cdot l_n < 0 \mid \psi} .
\end{align*}
Note that the probability in the right-hand side of the last equation
represents the probability that the first vote goes to the loser,
conditional on $\psi$. This turns out to be exactly
$\frac{n_B}{n_A+n_B}$, so we conclude that
$\pr{\dsem{m}{s}}{\phi \mid \psi} = 2\cdot \frac{n_B}{n_A+n_B}$
or equivalently
$\pr{\dsem{m}{s}}{\neg\phi \mid \psi} = \frac{n_A-n_B}{n_A+n_B}$
as desired.

We now turn to the proof of the \Sprhl judgments. By symmetry it
suffices to consider the first judgment. Using the rule of consequence
and the elimination rule for universal quantification, it suffices to
prove for every $i$:
$$\Equiv{s}{s}{\top}{l_1 \sidel \cdot l_n\sidel > 0  \land
  l_i\sidel =0 \land \psi\sidel 
  \Rightarrow  l_1 \sider \cdot l_n\sider < 0 \land l_i\sider =0 \land \psi\sider}
$$
We couple the samplings of $x$ using the negation function until
$|l| =i$, and then with the identity bijection. This establishes the
following loop invariant, from which we can conclude:
\[(\forall j\leq i.~l_j\sidel = - l_j\sider) \land l_i\sidel=l_i\sider=0 \land
    (\forall j> i.~l_j\sidel =l_j\sider) . \]

\section{Independence}\label{sec:indep}

We now turn to characterizing probabilistic independence using
couplings. We focus on probabilistic independence of program
variables, a common task when reasoning about randomized computations. In our
setting, the textbook definition of probabilistic independence can be cast as
follows:
\begin{definition}
A set $X=\{ x_1,\ldots, x_n\}$ of program variables of types
$A_1,\ldots, A_n$ is probabilistically independent in a
distribution $\mu\in\Dist(\Mem)$ iff for every $(a_1,\ldots, a_n) \in
A_1\times \ldots\times A_n$:
\[ \pr{\mu}{\bigwedge_{1\leq i \leq n} x_i = a_i} =
  \prod_{1\leq i \leq n} \pr{\mu}{x_i = a_i} . \]
\end{definition}

\subsection{Characterization}
Our first characterization of independence is based on the observation
that uniformity entails independence.
\begin{fact}[Independence from uniformity]
From every state $m$, if $X$ is uniformly distributed in $\dsem{m}{s}$ then $X$
is independent in $\dsem{m}{s}$.
\end{fact}
This observation enables proving independence by
coupling, in the special case where variables are uniform and
independent. For the general case, we will use an alternative characterization
based on self-composition.

\begin{prop}[Independence by coupling] \label{prop:indep-couple}
The following are equivalent:
\begin{enumerate}
\item for every state $m$, $X$ is independent in $\dsem{m}{s}$;
\item the following judgment, between a single copy of the program on the one 
  hand and a $n$-fold copy on the other hand, is derivable for every 
  $(a_1,\ldots, a_n)\in A_1\times\ldots \times A_n$:
  $$
  \Equiv{s}{\selfcomp{s}{n}}{\Eqmem{1}{n}}{
    \bigwedge_{1\leq i\leq n} x_{i} \sidel = a_i \iff
      \bigwedge_{1\leq i\leq n} x_i^i \sider = a_i} .
  $$
\end{enumerate}
\end{prop}
\begin{proof}
The validity of the universally quantified \Sprhl\ judgment is
equivalent to the following statement: for every $a_1,\ldots, a_n$ and
$n$-fold copy $\selfcomp{m}{n}$ of some initial state $m$,
$$\pr{\dsem{m}{s}}{\bigwedge_{1\leq i\leq n} x_{i} = a_i} =
  \pr{\dsem{\selfcomp{m}{n}}{\selfcomp{s}{n}}}{\bigwedge_{1\leq
      i\leq n} x_{i}^i = a_i} =
  \prod_{1\leq i\leq n} \pr{\dsem{m}{s}}{x_{i} = a_i}.$$
The last equality comes from the property of $n$-fold 
self-composition (\cref{prop:self-comp}).
\end{proof}

\subsection{Pairwise Independence of Bits}
\label{sec:pwindep}

Our first example is a well-known algorithm for generating $2^n$
pairwise independent bits. The algorithm first samples $n$ independent
bits $b_1\ldots b_n$, and then defines for every subset $X\subseteq
\{1,\ldots, n\}$ the bit $z_X =\bigoplus_{i\in X} b_i$.

\begin{wrapfigure}{l}{5.5cm}
  $$\boxed{
    \begin{array}{l}
    \ifor{i}{1}{n}{} \\
    \quad \irnd {b_i} {\{0,1\}}; \\
    \ifor{j}{0}{2^n -1}{} \\
    \quad \iass {z_{j}} {\bigoplus_{k\in \mathsf{bits}(j)} b_k};
    \end{array}
  }
  $$
\caption{Pairwise independence}\label{fig:pwindep}
\end{wrapfigure}

We can prove pairwise independence of the computed bits, i.e.\, for
every $X\neq Y$, $z_X$ and $z_Y$ are independent. Since there are
$2^n$ subsets of $\{1,\ldots, n\}$, this gives us $2^n$ pairwise
independent bits constructed from $n$ independent bits. The algorithm
is encoded by the program $s$ in \cref{fig:pwindep}, where
$\mathsf{bits}$ maps $\{0,\ldots, 2^n-1\}$ to a subset in
$\mathcal{P}(\{1,\ldots, n\})$ of positions that are $1$ in the binary
representation, and $\ifor{i}{a}{b}{s}$ is usual syntactic sugar for
$\kwhile$ loop with an incrementing counter $i$.

By our characterization based on self-composition, pairwise independence of
$z_j$ and $z_{j'}$ for every $j\neq j'$ is equivalent to the
(universally quantified) \Sprhl judgment
$$\Equiv{s}{s_1;s_2}{\top}{z_{j}\sidel= a \land
  z_{j'}\sidel =a' \iff z_{j}^1\sider= a \land
  z_{j'}^2\sider =a'} .$$
Since $j\neq j'$, the two sets $\mathsf{bits}(j)$ and $\mathsf{bits}(j')$
must differ in at least one element.  Let $k_0$ be the smallest element in 
which they differ.  Without loss of generality, we can assume that
$k_0\notin\mathsf{bits}(j)$ and $k_0\in \mathsf{bits}(j')$.
The crux of the proof is to establish the following judgment:
$$\Equiv{s_l}{s_r}{\top}{z\sidel= a \land
  z'\sidel =a' \iff z\sider= a \land
  z''\sider =a'}$$
where $z=\bigoplus_{k\in \mathsf{bits}(j)} b_k$, $z'=\bigoplus_{k\in
  \mathsf{bits}(j')} b_k$ and $z''=\bigoplus_{k\in \mathsf{bits}(j')}
b'_k$ and
\begin{align*}
  s_l & \eqdef
  \iforin{i}{[1\ldots n]\setminus k_0}{\irnd {b_i} {\{0,1\}}};~
  \irnd{b_{k_0}}{\{0,1\}} \\
  s_r & \eqdef
  \iforin{i}{[1\ldots n]\setminus k_0}
   {(\irnd {b_i} {\{0,1\}}; \irnd {b'_i} {\{0,1\}})};~
   \irnd {b_{k_0}} {\{0,1\}};~ \irnd {b'_{k_0}} {\{0,1\}} .
\end{align*}
This is proved by coupling the variables of the two programs in an
appropriate way.  We couple the random samplings as follows:

\begin{itemize}
\item for every $k \neq k_0$, we couple $b_k\sidel$ and
      $b_k\sider$ using the identity sampling;
\item we use the \textsc{Rnd-R} rule for $b'_k\sider$ for every $k \neq k_0$;
\item we couple $b_{k_0}\sidel$ and $b'_{k_0}\sider$ with the
  bijection which ensures
  $$b_{k_0}\sidel \oplus \left(\bigoplus_{k\in \mathsf{bits}(j') \setminus \{k_0\}} b_k\sidel\right) =
  b'_{k_0}\sider \oplus \left(\bigoplus_{k\in \mathsf{bits}(j') \setminus \{k_0\}} b'_k\sider\right) .
  $$
\end{itemize}
Putting everything together, the final proof obligation follows from the
algebraic properties of $\oplus$.


\subsection[k-wise Independence]{$k$-wise Independence}
\label{sec:kwindep}

The previous example can be generalized to achieve $k$-wise
independence for general $k$. Suppose we wish to generate $n$ random
variables that are $k$-wise independent. We will work in
$\mathbb{Z} /p\mathbb{Z}$, the field of integers modulo a prime $p$,
such that $k\leq p$.  Let $a_0, \dots, a_{k - 1}$ be drawn uniformly
at random from $\mathbb{Z}/p\mathbb{Z}$ and define the family of
random variables for every $m \in \{1,\dotsc, n\}$:
\[
  x_m = \sum_{j = 0}^{k - 1} a_j \cdot m^j ,
\]
where we take $0^0 = 1$ by convention. The corresponding code is given in 
\cref{fig:kwindep}. Then, we can show that any
collection of $k$ distinct variables $\{ x_i \}_i$ is independent.

\begin{wrapfigure}[9]{l}{6cm}
  $$\boxed{
    \begin{array}{l}
    \ifor{i}{1}{n}{} \\
    \quad \irnd {a_i} {\mathbb{Z}/p\mathbb{Z}}; \\
    \ifor{m}{0}{n-1}{} \\
    \quad \iass {x_m}{0} ;\\
    \quad \ifor{j}{0}{k-1}{}\\
    \quad \quad \iass {x_{m}} {a_j\cdot m^j};
    \end{array}
  }
  $$
\caption{$k$-wise independence}\label{fig:kwindep}
\end{wrapfigure}

For simplicity, we will show that the first $k$ elements $x_0, \dots,
x_{k-1}$ are uniform, and hence independent. Let $v_0,\ldots,v_{k -
  1}\in \Zn{n}$ be arbitrary elements of the field. Then the
probability that $(x_0,\ldots,x_{k - 1}) =(v_0,\ldots, v_{k - 1})$ is
equal to $p^{-k}$. Indeed, the equations
\[
\left\{ \begin{array}{@{}r@{~}c@{~}l}
v_0 & = & a_0  \\
\vdots & & \vdots \\
v_{k - 1} & = & \sum_{j = 0}^{k - 1} a_j \cdot (k - 1)^j
\end{array} \right. \]
define a system of linear equations with variables $a_0, \dots, a_{k-1}$.  By
basic linear algebra the system of equations has a unique solution for the
variables $a_0, \dots, a_{k - 1}$,\footnote{%
  Let the Vandermonde matrix $V(1,\dotsc, k-1)=
  \left(i^{j-1}\right)_{i,j}$ be
  $$
  V(1,\dotsc, k-1) \cdot (a_0, \dotsc, a_{k-1})^T = (v_0, \dotsc, v_{k-1})^T.
  $$
  The system of equations has a unique solution if and only if the matrix
  $V(1,\dotsc, k-1)$ is invertible in the space of matrices over $\Zn{n}$, which
  happens if and only if its determinant is non-zero mod $n$.  Expanding,
  $$
  \textrm{det}(V(1,\dotsc, k-1)) = \prod_{i\neq j} (i-j) = \prod_{i=2}^{k-1} 
  i!
  $$
  Note that $p$ does not divide the determinant by Gauss' lemma, since $n$ can't
  divide any of the terms $i!$ for any $i<n$. Therefore, the system of equations
  has a unique solution.}
which we denote $(v_0^*,\ldots, v_{k - 1}^*)$. Now consider the \Sprhl judgment
that establishes uniformity:
$$
\Equiv{s}{s}{\top}{x_0\sidel= v_0 \land \cdots \land x_{k - 1} \sidel= v_{k 
    -
    1}
  \iff x_0\sider= w_0 \land \cdots \land x_{k - 1} \sider= w_{k - 
1}}
$$
By applying (relational) weakest precondition on the deterministic
fragments of the program, the judgment is reduced to
$$\Equiv{s}{s}{\top}{a_0\sidel= v_0^* \land \cdots \land a_{k-1}
  \sidel= v_{k - 1}^* \iff a_0\sider= w_0^* \land \cdots \land a_{k-1}
  \sider= w_{k - 1}^*}$$
We then repeatedly apply the rule for random sampling, with the
permutation on $\mathbb{Z}/p\mathbb{Z}$ that exchanges $(v_i^*,
w_i^*)$.

\section{Conditional Independence}\label{sec:cindep}

Finally, we consider how to show conditional independence. Recall that the
conditional probability $\pr{x\sim\mu}{A\mid B}$ is defined when
$\pr{x\sim\mu}{B} \neq 0$ and satisfies $\pr{x\sim\mu}{A\mid B} \eqdef
\frac{\pr{x\sim \mu}{A \land B}}{\pr{x\sim\mu}{B}}$.
\begin{definition}
Let $X=\{ x_1,\ldots, x_n\}$ be a set of program variables of types
$A_1,\ldots, A_n$ and let $E$ be an event. We say that $X$ is
independent conditioned on $E$ in a
distribution $\mu\in\Dist(\Mem)$ iff for every $(a_1,\ldots, a_n) \in A_1\times \ldots
\times A_n$:
$$\pr{\mu}{\left. \bigwedge_{1\leq i \leq n} x_i =  a_i ~\right|~ E} =
\prod_{1\leq i \leq n} \pr{\mu}{x_i = a_i \mid E} .$$
(For this definition to make sense, we are implicitly assuming that
$\pr{\mu}{E}\neq 0$.)
\end{definition}

The following lemma unfolds the definition of conditional independence
and is useful for the characterization of the next section.
\begin{lemma}
A set of variables $X$ is independent conditioned on an event $E$
in $\mu$ iff for every $a_1\in A_1$, \ldots, $a_n\in A_n$:
  $$
  \pr{\mu}{\bigwedge_{1\leq i \leq n} x_i = a_i\land E}
  \cdot \left(\pr{\mu}{E}\right)^{n-1} 
  = \prod_{1\leq i \leq n} \pr{\mu}{x_i = a_i \land E} .
 $$

\end{lemma}

\subsection{Characterization}

The characterization of independence based on self-composition can
be extended as follows.

\begin{prop}[Conditional independence by coupling]
The following are equivalent:
\begin{enumerate}
\item for every state $m$, $X$ is independent conditioned on $E$
  in $\dsem{m}{s}$;
\item $\Equiv{\selfcomp{s}{n}}{\selfcomp{s}{n}}{\Eqmem{n}{n}}{
  \left( \phi_1
    \land \mathbb{E}\sidel \right)
  \iff
  \left( \phi_2
    \land \mathbb{E}\sider \right)}$ 
  where
  $\mathbb{E} \eqdef \bigwedge_{1\leq i\leq n} E^i$,
  $\phi_1 \eqdef \bigwedge_{1\leq i\leq n} (x_i^1 \sidel = a_i)$ and
  $\phi_2 \eqdef \bigwedge_{1\leq i\leq n} (x^i_i \sider = a_i)$.
\end{enumerate}
\end{prop}

\begin{proof}
The proof is similar to the case of independence (\cref{prop:indep-couple}).
\end{proof}

\subsection{Example: Conditional Independence}

\begin{wrapfigure}{l}{4.5cm}
  $$\boxed{
    \begin{array}{l}
    \irnd x \mu; \\
    \irnd y \mu'; \\
    \irnd z \mu''; \\
    \iass w {f(x,y)}; \\
    \iass {w'} {g(y,z)};
    \end{array}
  }
  $$
\caption{Conditional indep.}\label{fig:condindep}
\end{wrapfigure}

We consider a simple example often used to illustrate Bayesian
networks models. Let $x$, $y$, $z$, $w$ and $w'$ be random variables,
where $x$, $y$ and $z$ are sampled from distributions $\mu$, $\mu'$
and $\mu''$ respectively, and $w$ and $w'$ are defined by their
respective assignments. Both $w$ and $w'$ depend on $y$, along with
independent sources of randomness, respectively $x$ and $z$. While $w$
and $w'$ are not independent---they share dependence on $y$---if we
\emph{condition} on a particular value of $y$, then $w$ and $w'$ are
independent.

The code of the corresponding program $s$ is given in 
\cref{fig:condindep}. We
want to show that $w$ and $w'$ are independent conditioned on $y=c$
for every $c$. Using our characterization based on self-composition,
it amounts to proving the following (universally quantified)
\Sprhl judgment:
\[
  \Equiv{\selfcomp{s}{2}}
    {\selfcomp{s}{2}}{\Eqmem{2}{2}}
    {\phi\sidel \iff \psi\sider} \]
where
$\left\{ \begin{aligned}
  \phi & \eqdef  w^1 = a \land {w'}^1= b \land y^1 = c  \land y^2 =c \\
  \psi & \eqdef w^1 = a \land {w'}^2= b \land y^1 = c \land y^2 = c .
\end{aligned} \right.$

The proof proceeds by moving the samplings of $z^1$ and $z^2$ in both programs
to the front of the program, and then swapping samplings in the left program (we
can use the rule \rname{Swap} to reorder the instructions, as the sampling
instructions for $z^1$ and $z^2$ operate on different variables). Then, we
couple $z^1\sidel$ to be equal to $z^2\sider$, and $z^2\sidel$ to be equal to
$z^1\sider$. We apply the identity coupling to all other random samplings.

\section{Formalization}
\EasyCrypt~\citep{Barthe:2011:CRYPTO,FOSAD:BDGKS14} is an interactive
proof assistant that supports reasoning about (relational) properties
of probabilistic programs, using the \Sprhl\ logic. We have applied
\EasyCrypt to the main examples of this paper. For uniformity, it
suffices to establish the required \Sprhl\ judgment; in contrast,
independence via self-composition requires to build the self-composed
program, which we have done manually. The main challenges for the
verification are:
\begin{enumerate}
\item Restructuring the code of the program to make the rules of the
  logic applicable; this is done by applying the equivalent of the
  \rname{Struct} rules.

\item Discovering and establishing the correct proof invariants. The
  current version of \EasyCrypt requires that invariants are produced
  by the users.

\item Building an appropriate coupling, primarily through applying the
  rule for random samplings with carefully chosen bijections.
\end{enumerate}
Our examples are formalized in about 1,000 lines of proof script in
\EasyCrypt.\footnote{%
  Proofs are available at the following link:
  \url{https://gitlab.com/easycrypt/indep}}
The most complex example, and the one where the three challenges are most
pronounced, is the random walk over a cycle. This example is formalized in about
500 lines of \EasyCrypt code, out of which the statement, including the
definition of the program, takes about 50 lines. The remaining 90\% of the
formalization covers the notions used in the proof and the proof itself. 

\section{Conclusion}
We have proposed a new method based on probabilistic couplings for formally
verifying uniformity and independence properties of probabilistic programs. Our
method complements the existing range of techniques for probabilistic reasoning,
and has many potential applications in program verification, security, and
privacy.

\section{Acknowledgments}

We thank the anonymous reviewers for their detailed comments.  This work was
partially supported by NSF grants TC-1065060 and TWC-1513694, by the European
Union's H2020 Programme under grant agreement number ICT-644209, and a
grant from the Simons Foundation ($\#360368$ to Justin Hsu).
\bibliographystyle{abbrvnat} \bibliography{header,main}

\end{document}
